\documentclass[11pt, oneside, a4paper]{article}

\usepackage[left=3.2cm,right=3.2cm,top=3cm]{geometry}
\usepackage[utf8]{inputenc}
\usepackage[T1]{fontenc}
\usepackage{hyperref}
\usepackage{microtype}
\usepackage{xcolor}

\usepackage{mathtools, amsthm, amsfonts}

\usepackage{booktabs}
\usepackage{braket}
\usepackage{commath}
\usepackage{caption}
\usepackage{titlesec}
\titlelabel{\thetitle. }
\titleformat*{\section}{\large\bfseries}

\usepackage[capitalise]{cleveref}

\usepackage{tikz}
\usetikzlibrary{positioning}

\newtheorem{theorem}{Theorem}
\newtheorem{proposition}[theorem]{Proposition}
\newtheorem{lemma}[theorem]{Lemma}

\theoremstyle{definition}
\newtheorem{definition}[theorem]{Definition}

\newtheorem*{problem*}{Problem}

\newtheorem{remark}[theorem]{Remark}

\DeclarePairedDelimiter\ceil{\lceil}{\rceil}

\DeclareMathOperator{\supprank}{R_s}
\DeclareMathOperator{\bsupprank}{\underline{R}_s}
\DeclareMathOperator{\rank}{R}
\DeclareMathOperator{\borderrank}{\underline{R}}

\newcommand{\Oh}{\mathcal{O}}

\newcommand{\CC}{\mathbf{C}}
\newcommand{\NN}{\mathbf{N}}

\newcommand{\GHZ}{\mathrm{GHZ}}

\newcommand{\RR}{\mathbf{R}}

\newcommand{\IMM}{\mathrm{IMM}}

\newcommand{\EQ}{\mathrm{EQ}}

\DeclareMathOperator{\NQ}{\mathrm{NQ}}
\DeclareMathOperator{\N}{\mathrm{N}}
\DeclareMathOperator{\borderNQ}{\underline{\NQ}}

\newcommand{\ketbra}[2]{\ket{#1}\!\!\bra{#2}}

\DeclareMathOperator{\Id}{Id}
\DeclareMathOperator{\tr}{tr}

\DeclareMathOperator{\Cyc}{Cyc}
\DeclareMathOperator{\Sym}{Sym}

\newcommand{\VP}{\mathrm{VP}}
\newcommand{\VNP}{\mathrm{VNP}}
\newcommand{\VPe}{\mathrm{VP}_{\!\mathrm{e}}}
\newcommand{\VQP}{\mathrm{VQP}}

\DeclareMathOperator{\rk}{rk}

\newcommand{\Str}{\mathrm{Str}}

\begin{document}

\vspace*{2em}
\begin{center}
\Large\textbf{Nondeterministic quantum communication complexity:\\[0.2em] the cyclic equality game and iterated matrix multiplication}\par
\vspace{1em}
\large Harry Buhrman\footnote{QuSoft, CWI Amsterdam and University of Amsterdam, Science Park 123, 1098 XG Amsterdam, Netherlands. Email: buhrman@cwi.nl, j.zuiddam@cwi.nl}, Matthias Christandl\footnote{Department of Mathematical Sciences, University of Copenhagen, Universitetsparken 5,  2100~Copenhagen Ø, Denmark. Email: christandl@math.ku.dk}, Jeroen Zuiddam${}^1$\par
\vspace{1em}
\today
\end{center}
\vspace{1em}

\begin{abstract}
We study nondeterministic multiparty quantum communication with a quantum generalization of broadcasts. We show that, with number-in-hand classical inputs, the communication complexity of a Boolean function in this communication model equals the logarithm of the \emph{support rank} of the corresponding tensor, whereas the approximation complexity in this model equals the logarithm of the \emph{border support rank}. This characterisation allows us to prove a log-rank conjecture posed by Villagra et~al.\ for nondeterministic multiparty quantum communication with message passing.

The support rank characterization of the communication model connects quantum communication complexity intimately to the theory of asymptotic entanglement transformation and algebraic complexity theory. In this context, we introduce the \emph{graphwise equality problem}. For a cycle graph, the complexity of this communication problem is closely related to the complexity of the computational problem of multiplying matrices, or more precisely, it equals the logarithm of the support rank of the iterated matrix multiplication tensor. We employ Strassen's laser method to show that asymptotically there exist nontrivial protocols for every odd-player cyclic equality problem. We exhibit an efficient protocol for the 5-player problem for small inputs, and we show how Young flattenings yield nontrivial complexity lower bounds.
\end{abstract}

\section{Introduction}

Let $f : X\times Y \times Z \to \{0,1\}$ be a function on a product of  finite sets $X$, $Y$ and $Z$. Alice, Bob and Charlie have to compute~$f$ in the following sense. Alice receives an $x\in X$, Bob receives a $y\in Y$ and Charlie receives a $z\in Z$, and each player receives a private random bit string. Then the players communicate in rounds. Each round, one player communicates by broadcasting a bit to the other players. After these rounds of communication, each player has to output a bit, such that if $f(x,y,z) = 1$, then with some nonzero probability all players output 1 and if $f(x,y,z)=0$, then with probability zero all players output 1. The complexity of such a protocol is the number of broadcasts in the protocol, and we denote the minimum complexity of all such protocols by $\N(f)$.

Now we allow the players to be quantum, as follows. Alice receives an $x\in X$, Bob receives a $y\in Y$ and Charlie receives a $z\in Z$. Then, in rounds, the players communicate by creating a $\GHZ$-like state
\[
\ket{\GHZ} = \alpha\ket{000} + \beta\ket{111}
\]
and sharing this state among each other, a quantum broadcast. Moreover, the players can do any local quantum computations. Again, after these rounds of communication, each player has to output a bit, such that if $f(x,y,z) = 1$, then with some nonzero probability all players output~1 and if $f(x,y,z)=0$, then with probability zero all players output 1. The quantum complexity of such a quantum protocol is the number of broadcasts in the protocol, and we denote the minimum complexity of all quantum protocols by $\NQ(f)$. We will make this definition more precise and more general in \cref{model}. Note that the quantum model can simulate the classical model. Also note that, nondeterministically, one quantum broadcast can be used to send a qubit from one player to another by using teleportation (see \cref{logrank}); the quantum model can thus simulate a message-passing model.

\paragraph{Our results.} 
\begin{itemize}
\item
Our main technical result is that the quantum complexity of a function in the above model equals the logarithm of the so-called \emph{support rank} of the tensor $\sum_{x,y,z} f(x,y,z)\ket{x}\!\ket{y}\!\ket{z}$ corresponding to $f$. We prove this in \cref{model}.
\item Modifying the quantum model such that the players can only communicate by message passing --- that is, in each communication round one player sends a qubit to one other player --- increases the complexity by at most a factor $k-1$, and this relationship is tight. However, asymptotically in the input size, the increase is only $k/2$ and this relationship is tight. This solves a \emph{nondeterministic multiplayer quantum log-rank conjecture} in the message-passing model of Villagra et al.~\cite{villagra2013tensor}. This topic is covered in \cref{logranksec}.
\item We define the $k$-player \emph{graphwise equality} problem to be the problem in which $k$ players are identified with vertices in a graph $G$, and each player has to compute the equality function with his neighbours in $G$. Of particular interest is the cycle graph $G = C_k$ and the corresponding \emph{cyclic equality} problem. For this cyclic equality problem, in the classical broadcast model, the naïve protocol in which every player broadcasts his inputs is the optimal protocol. The same holds in the quantum model when $k$ is even.  Interestingly, we show  with Strassen's laser method that for all odd $k \geq 3$ there is a nontrivial quantum protocol. Moreover, for all odd $k\geq 3$ we give nontrivial lower bounds on the value of $\NQ$ by use of Young flattenings. These results are related to the complexity of matrix multiplication and iterated matrix multiplication. A consequence of our work is that finding new protocols for the cyclic equality problem for three players yields new algorithms for matrix multiplication. \cref{cep} covers the classical case, the even quantum case, an explicit quantum protocol for $k=5$, and the Young flattening lower bound. \cref{lasersec} covers the Strassen laser method.
\end{itemize}

\paragraph{Related work.} The two-player nondeterministic quantum communication model was introduced by~De Wolf~\cite{de2003nondeterministic}. He shows that the communication complexity in this model is characterized by the logarithm of the support rank of the communication matrix. Quantum broadcast channels have been studied by e.g.~Ambainis~et~al.~\cite{ambainis2004multiparty}.
Multiparty nondeterministic quantum communication with message passing has been studied by Villagra et~al.~\cite{villagra2013tensor}. They show that the logarithm of the support rank of the communication tensor is a lower bound for the message-passing complexity and conjecture that this lower bound is polynomially related to the message-passing complexity.

The support rank of 3-tensors has been studied by Cohn and Umans in the context of the complexity of matrix multiplication \cite{cohn2013fast}. They give nontrivial upper bounds on the support rank of the matrix multiplication tensor that do not come from upper bounds on the tensor rank. As an interesting fact, we note that given a matrix $A$ and a number $k$, deciding whether the support rank of $A$ is at least $k$ is NP-hard \cite{bhangale2015complexity}.

The complexity of matrix multiplication plays a central role in algebraic complexity theory. We refer to \cite{burgisser1997algebraic} for general background information. Connections between algebraic complexity theory and entanglement transformations have been studied before, see  for example \cite{chitambar2008tripartite}. The iterated matrix multiplication tensor has been studied in the context of arithmetic circuit complexity and the $\VP$ versus $\VNP$ problem, see for example \cite{gesmundo2015gemetric}. To the knowledge of the authors, the tensor rank or support rank of the iterated matrix multiplication tensor has not been studied before.

\paragraph{Acknowledgements.} We thank Peter Bürgisser, Péter Vrana, Florian Speelman and Teresa Piovesan for helpful discussions. Part of this work was done while MC and JZ were visiting the Simons Institute for the Theory of Computing, UC Berkeley. HB was partially funded by the European Commission, through the SIQS project and by the Netherlands Organisation for Scientific Research (NWO) through gravitation grant Networks. MC acknowledges financial support from the European Research Council (ERC Grant Agreement no 337603), the Danish Council for Independent Research (Sapere Aude) and the Swiss National Science Foundation (project no PP00P2\_150734). Part of this work was done while MC was with ETH Zurich. JZ is supported by NWO through the research programme 617.023.116 and by the European Commission through the SIQS project.

\section{Support rank characterization of the quantum broadcast model}\label{model}

We refer to Nielsen and Chuang \cite{nielsen2010quantum} for background information on the quantum computation model.

\paragraph{Quantum multiparty communication protocol.} We will give two definitions of a quantum broadcast model, which are equivalent in the nondeterministic setting. The first model clearly generalizes the classical broadcast model, while the second model is easier to analyse.
For any natural number~$m$, denote by $[m]$ the set $\{1,2, \ldots, m\}$. Let $k$ be a positive integer and let $f$ be a Boolean function on $[2^n]^{k} = [2^n] \times [2^n] \times \cdots\times [2^n]$, 
\[
f: [2^n]^k \to \{0,1\}.
\]

We define a \emph{$k$-player quantum communication protocol} as follows. Each player $i$ has a local Hilbert space $H_i$ with a register initialised in the input state $\ket{x_i}$. The players have access to a quantum broadcast channel, which, given a qubit state $\alpha\ket{0} + \beta\ket{1}$, will create the state $\alpha\ket{0}^{\otimes k} + \beta\ket{1}^{\otimes k}$ and distribute this state among the $k$ players. The players proceed in communication rounds; each round a designated player uses the broadcast channel.
Let $R_i$ be the first qubit of $H_i$ and let $R = R_1 \otimes \cdots \otimes R_k$. After the communication is finished, we apply a projection onto $\ket{11\cdots1}$ in $R$. If the resulting tensor is 0 then the output of the protocol is 0, otherwise the output of the protocol is~1. The complexity of the protocol is the number of communication rounds. 
We say the protocol nondeterministically computes $f$ if the probability that the output equals~1 is nonzero if $f(x_1,\ldots,x_k) = 1$ and the probability that the output equals 0 is one if $f(x_1, \ldots, x_k) = 0$.

We will now give an equivalent definition of the quantum broadcast model. This is the definition that we will use in the rest of the paper.
Each player $i$ has a finite-dimensional Hilbert space $H_i$. 
The protocol thus takes place in the space $H_1 \otimes \cdots \otimes H_k$. 
The space is initialised in the state $\ket{x_1\cdots x_k}\ket{\GHZ^k_r}$, where
\[
\ket{\GHZ^k_r} \coloneqq \sum_{a=1}^{\smash{r}} \ket{a}\!\ket{a} \cdots \ket{a} \in (\CC^{r})^{\otimes k}
\]
is the $k$-party $\GHZ$-state of rank $r$, shared among the $k$ players, and $x_i\in [2^n]$ is the classical input to player $i$. (For clarity we will suppress any normalizations in quantum states when possible.)  
The players now apply local quantum operations.
Let $R_i$ be the first qubit of $H_i$ and let $R = R_1 \otimes \cdots \otimes R_k$. We apply a projection onto $\ket{11\cdots1}$ in $R$. If the resulting tensor is 0 then the output of the protocol is 0, otherwise the output of the protocol is~1. The \emph{complexity} of the protocol is $\log_2(r)$. 
We say the protocol \emph{nondeterministically computes} $f$ if the probability that the output equals~1 is nonzero if $f(x_1,\ldots,x_k) = 1$ and the probability that the output equals 0 is one if $f(x_1, \ldots, x_k) = 0$.

\begin{definition}
Let $k$ be a positive integer and let $f$ be a function $[2^n]^k \to \{0,1\}$. The \emph{$k$-player nondeterministic quantum communication complexity of $f$} is the minimal complexity of a $k$-player quantum communication protocol that nondeterministically computes $f$, and is denoted by $\NQ(f)$.
\end{definition}

\paragraph{Approximating protocols.} Let $f$ be a function $[2^n]^k\to\{0,1\}$.  Let $(\Pi_j)_{j\in \NN}$ be a sequence of protocols, such that when $f(x_1,\ldots,x_k) = 1$, the probability that $\Pi_j$ outputs~1 on input $x$ converges to a nonzero number as $j$ goes to infinity, and when $f(x_1,\ldots, x_k) = 0$, the probability that $\Pi_j$ outputs 0 on input $x$ converges to 1 as $j$ goes to infinity. Then we say that the sequence $(\Pi_j)_{j\in \NN}$ \emph{approximately nondeterministically computes} $f$. The complexity of an approximating sequence is the maximum complexity of any protocol $\Pi_j$ in the sequence.

\begin{definition}
The \emph{$k$-player approximate nondeterministic quantum communication complexity of $f$} is the minimal complexity of a sequence $(\Pi_j)$ that approximately nondeterministically computes $f$, and is denoted by $\borderNQ(f)$.
\end{definition}

\paragraph{Classical protocol.} We define a \emph{$k$-player classical communication protocol} as follows. Each player receives a classical input and a private random bit string. The protocol proceeds in rounds. Each round we let a single predetermined player communicate by broadcasting a bit to all the other players. After the last communication round, every player presents an output bit. If all the output bits are 1, then the output of the protocol is 1; otherwise the output of the protocol is 0. Again, we say the classical protocol \emph{nondeterministically computes $f$} if the probability that the output equals 1 is nonzero if $f(x_1,\ldots,x_k) = 1$ and the probability that the output equals $0$ is one if $f(x_1, \ldots,x_k) = 0$.

\begin{definition}
The \emph{$k$-player nondeterministic classical communication complexity of $f$} is the minimal complexity of a $k$-player classical communication protocol that nondeterministically computes $f$, and is denoted by $\N(f)$.
\end{definition}

\begin{remark}
For simplicity, we have taken the input set for each of the $k$ players to be the same set~$[2^n]$. We note that the definitions in this section and most of the results in this paper naturally generalize to the situation where the players get inputs from sets of different sizes.
\end{remark}

\paragraph{Support rank and border support rank.}
Let $t$ be a tensor in $(\CC^m)^{\otimes k}$. The \emph{tensor rank} of $t$ is the smallest number $r$ such that~$t$ can be written as a sum of $r$ simple tensors, that is,  $t = \sum_{i=1}^r u^1_i \otimes u^2_i \otimes \cdots \otimes u^k_i$ for some vectors $u^j_i \in \CC^m$. We denote the tensor rank of $t$ by $\rank(t)$. Fix a basis for $(\CC^m)^{\otimes k}$ and define the support of a tensor $t$ in $(\CC^m)^{\otimes k}$ to be the set of basis element that occur with nonzero coefficient in $t$.  The \emph{support rank} or \emph{nondeterministic rank} of $t$ is the smallest number $r$ such that there exists a tensor in the space $(\CC^m)^{\otimes k}$ with the same support as $t$ and tensor rank $r$. We denote the support rank of $t$ by $\supprank(t)$. Note that support rank is basis \emph{dependent}.

The \emph{border rank} of $t$ is the smallest number $r$ such that there exists a sequence of tensors $(t_j)_{j\in \NN}$ converging to $t$ in the Euclidean topology (or equivalently in the Zariski topology) such that $\rank(t_j)$ is at most $r$ for every $j$. We denote the border rank of $t$ by~$\borderrank(t)$. The \emph{border support rank} of $t$ is the smallest number $r$ such that there exists a tensor in $(\CC^m)^{\otimes k}$ with the same support as $t$ and border rank $r$. We denote the border support rank of $t$ by $\bsupprank(t)$.

\begin{theorem}\label{mainth}
Let $f: [2^n]^k \to \{0,1\}$ be a function and let $t$ be the tensor in $(\CC^{2^n})^{\otimes k}$ with entries given by $f$, that is, $t = \sum_{i\in [2^n]^k} f(i) \ket{i_1}\!\ket{i_2}\cdots \ket{i_k}$. Then $\NQ(f) = \log_2 \supprank(t)$ and $\borderNQ(f) = \log_2 \bsupprank(t)$.
\end{theorem}

\begin{lemma}[Cleanup lemma]\label{cleanup} Let $\{\ket{\psi_i}  : i\in [q]\}\subseteq (\CC^m)^{\otimes k}$ be a set of $k$-tensors, for some natural number $q$. Then there exists a $k$-partite rank-1 linear map $\bra{\ell}\coloneqq\bra{\ell_1}\otimes \cdots \otimes \bra{\ell_k}$ with $\bra{\ell_j}\in (\CC^m)^*$ such that $\braket{\ell|\psi_i}\neq 0$ for every $i\in[q]$.
\end{lemma}
\begin{proof}
We will give a proof by recursively constructing $\bra{\ell}$. Let $\Id$ be the identity map on $\CC^m$. If $j\leq k$,  $\bra{a}\in ((\CC^m)^*)^{\otimes j}$ and $\ket{b} \in (\CC^m)^{\otimes k}$, then we denote by $\braket{a|b}$ the contraction of $\bra{a}$ and $\ket{b}$, that is, $\braket{a|b} = (\bra{a} \otimes \Id^{\otimes k-j})\ket{b}$. 

The base case is $\bra{\ell} = 1$. For the recursion, suppose we are given an element $\bra{\ell'} \in ((\CC^m)^*)^{\otimes j}$ such that $\ket{\phi_i}\coloneqq \braket{\ell|\psi_i}$ is nonzero for every $i\in [q]$. We will construct an element $\bra{\ell}\in((\CC^m)^*)^{\otimes j+1}$ such that $\braket{\ell|\psi_i}$ is nonzero for every $i\in [q]$. Since $\ket{\phi_i}$ is nonzero for every $i\in [q]$, there is an element $\bra{u_i}\in (\CC^m)^*$ such that $\braket{u_i|\phi_i}$ is nonzero. Consider the the maps $(\bra{u_1} + x \bra{u_2})\ket{\phi_i}$ for $i\in \{1,2\}$, in the variable $x$. Each map only has a single root. Therefore, there exists a value $\alpha_2$ for $x$ such that both maps evaluate to a nonzero number. Next, consider the maps $(\bra{u_1} + \alpha_2 \bra{u_2} + x\bra{u_3})\ket{\phi_i}$ for $i\in \{1,2,3\}$, in variable $x$. Again, each of the three maps has only a single root. Therefore, there exists a value $\alpha_3$ for $x$ such that all three maps evaluate to a nonzero number. Repeat this construction to obtain an element $\bra{u} \in (\CC^m)^*$ such that $\braket{u|\phi_i}$ is nonzero for every $i\in[q]$. Let $\bra{\ell}$ be $\bra{\ell'}\otimes \bra{u}$.
\end{proof}

\begin{proof}[\bfseries\upshape Proof of \cref{mainth}]
We first show $\NQ(f) \leq \log_2 \supprank(t)$. Let $r$ be the support rank of $t$. Then there exists a unit vector $\psi \in (\CC^{2^n})^{\otimes k}$ with rank $r$ and support equal to the support of $f$. This means that there are vectors $\ket{u^j_i} \in \CC^{2^n}$ such that $\psi = \sum_{i=1}^r \ket{u^1_i}\cdots \ket{u^k_i}$. For every player $j$ define a matrix
\[
A_j \coloneqq \alpha_j \sum_{i=1}^{\smash{r}} \ketbra{u^j_i}{i}
\]
where $\alpha_j$ is a nonzero complex number such that $\smash{A_j^\dag A_j}$ has eigenvalue at most 1. The matrix $I - \smash{A_j^\dag A_j}$ is thus positive semidefinite and hence there exists a matrix $A_j'$ such that $\smash{{A_j'}^\dag} A_j' = \smash{I-A_j^\dag A_j}$. Define for every player $j$ a quantum operation
\[
\mathcal{E}_j : \rho \mapsto A_j\rho A_j^\dag \otimes \ketbra{1}{1} + A_j' \rho {A_j'}^\dag \otimes \ketbra{0}{0}.
\]
Note that this operation introduces a new control qubit register which player $j$ can measure to see whether he applied $A_j$ or $A'_j$.

The protocol for the $k$ players is as follows. Let $x_1, \ldots, x_k$ be the inputs given to the players. The players share a $k$-party $\GHZ$-state of rank $r$. Player $j$ applies $\mathcal{E}_j$ to his part of the $\GHZ$-state. If his control qubit is $\ket{0}$ then he sets his output qubit $R_i$ to $\ket{0}$. Otherwise, he measures the rest of the system. If the outcome equals $\ket{x_j}$, then he sets $R_j$ to $\ket{1}$, otherwise he sets $R_j$ to $\ket{0}$.

The above protocol uses a $\GHZ$-state of rank $r$, so it has complexity $\log_2(r)$. We claim that the protocol nondeterministically computes $f$. If the players in the first measurement each get outcome $\ket{1}$, then the state of the total system is $\ket{\psi}$. Because $\ket{\psi}$ has norm 1, this happens with nonzero probability~$\abs{\alpha_1}^2\cdots\abs{\alpha_k}^2$. If $f(x_1,\ldots,x_k) = 0$, then $\ket{x_1\cdots x_k}$ does not occur in the support of $\psi$, so the probability that the players measure $\ket{x_1}, \ldots, \ket{x_k}$ respectively is zero. Hence in this case the register $R$ is not in state $\ket{11\cdots 1}$. On the other hand, if~$f(x_1, \ldots, x_k) \neq 0$, then $\ket{x_1\cdots x_k}$ does occur in the support of $\psi$, so the probability that the players measure $\ket{x_1},\ldots,\ket{x_k}$ respectively is nonzero. Hence with nonzero probability the register $R$ is in state $\ket{11\cdots1}$.

We now show $\NQ(f) \geq\log_2 \supprank(t)$. Suppose we have a protocol that nondeterministically computes $f$ with complexity~$r$. This means that the players perform \emph{local} quantum operations that together form a linear map $L$ which  transforms, for any $x_1,\ldots,x_k\in [2^n]$, the state
\begin{align*}
&\ket{x_1\cdots x_k}\ket{\GHZ_r}\\
\shortintertext{to a state of the form}
&\ket{x_1\cdots x_k} \sum_{a\in A}\, \ket{\psi_x^a}\,\ket{a_1}\!\ket{a_2}\cdots \ket{a_k},
\end{align*}
where the sum is over $A\coloneqq\{a\in \{0,1\}^k\mid f(x_1,\ldots,x_k) = a_1\cdot a_2\cdots a_k\}$ and where $\ket{\psi_x^a}$ is some nonzero vector, representing the state of the \emph{work space} of the players. Since the map~$L$ is \emph{linear}, it maps the tensor
\begin{align*}
&s_1\coloneqq \sum_{\mathclap{x_1,\ldots,x_k}} \ket{x_1\cdots x_k} \ket{\GHZ_r}\\
\shortintertext{to the tensor}
&s_2\coloneqq \sum_{\mathclap{x_1,\ldots,x_k}} \ket{x_1\cdots x_k} \sum_{a\in S} \ket{\psi_x^a}\ket{a_1\cdots a_k}.
\end{align*}
The tensor rank of $\sum_x \ket{x_1\cdots x_k}$ is 1 and hence the tensor rank of $s_1$ is $r$. Because~$L$ is a local map, the tensor rank of $s_2$ is at most $r$. By applying the cleanup lemma \cref{cleanup} and projecting on states with $\ket{a_1\cdots a_k} = \ket{1\cdots 1}$, we obtain a tensor
\[
s_3 \coloneqq \sum_{x_1,\ldots, x_k} \ket{x_1\cdots x_k} c_x
\]
where $c_x\in \CC$ is zero if $f(x) = 0$ and nonzero if $f(x) = 1$. The rank of the tensor $s_3$ is at most $r$. The support of $s_3$ equals the support of $f$, so the support rank of $f$ is at most~$r$.

The statement about the approximate complexity of $f$ follows from the definition of border support rank.
\end{proof}

\begin{remark}
We note that having a $\NQ$-protocol for $f$ of complexity $n$ is the same as having an SLOCC protocol for transforming $\GHZ_{2^n}^k$ to a tensor with the same support as $f$. We will use the SLOCC paradigm in some parts of the text.
\end{remark}

\section{Nondeterministic log-rank conjecture for message-passing protocols}\label{logranksec}

\begin{definition}
Let $\NQ_0(f)$ be the minimal complexity of a protocol that nondeterministically computes $f$, without preshared entanglement but with the added ability for players to send a qubit to another player. The complexity of such a protocol is the total number of qubits sent.
\end{definition}

Villagra et al.~\cite{villagra2013tensor} show that $\NQ_0(f)$ is at least the logarithm of the support rank of $f$. They furthermore conjecture that $\NQ_0(f)$ is upper bounded by a polynomial in the logarithm of the support rank. The following theorem proves this conjecture.


\begin{theorem}[``Nondeterministic log-rank conjecture'']\label{logrank} Let $f: [2^n]^k\to \{0,1\}$. Then we have
$\NQ(f) \leq \NQ_0(f) \leq (k-1) \NQ(f)$.
\end{theorem}
\begin{proof}
For the first inequality, suppose we have an $\NQ_0$-protocol for $f$. We replace the communication of a qubit by the nondeterministic teleportation of that qubit. Beforehand, all players agree on the basis in which the teleportation should happen. If any teleportation during the protocol does not happen in this basis, then the player that notices this sets his output register $R_i$ to $\ket{0}$.

For the second inequality, suppose we have an $\NQ$-protocol for $f$ which uses a $\GHZ$-state of rank $r$. Then we can construct a $\NQ_0$-protocol for $f$ as follows. The players start with no shared entanglement. Player~1 constructs a $\GHZ$-state of rank $r$ locally. In the first $k-1$ communication rounds, player~$1$ distributes the $\GHZ$-state over the other $k-1$ players. After that, the players perform the $\NQ$-protocol. The resulting $\NQ_0$-protocol has complexity at most $(k-1)\NQ(f)$.
\end{proof}

To say something about the `tightness' of \cref{logrank} we consider the natural easy function in the $\NQ$-model, namely $f(x_1,\ldots,x_k) = [x_1=x_2=\cdots=x_k]$ with $x_i\in [2^n]$.

\begin{proposition}[Single bit inputs]
Let $f:[2]^k\to\{0,1\}$ be the function defined by $f(x_1,\ldots,x_k) = [x_1=x_2=\cdots=x_k]$ for $x_i \in [2]$. Then we have $\NQ(f) = 1$ and $\NQ_0(f) = (k-1)\NQ(f)$.
\end{proposition}

\begin{proof}
 Note that the tensor of this function is $\GHZ^k_{2}$, so $\NQ(f) = 1$. Now consider a protocol that nondeterministically computes $f$ without preshared entanglement and $r$ rounds of communication. We may assume, without loss of generality, that the protocol consists of a first phase in which the players communicate and a second phase in which the players only do local quantum operations. After the first phase the players are sharing some state $E$ consisting of EPR-pairs shared among certain pairs of the players. We thus obtain a local linear map which maps $\sum_x \ket{x} E$ to a tensor with the same support as $\GHZ^k_2$. However, if $r < k-1$, then, viewing $E$ as a graph, $E$ is disconnected. Therefore there is a grouping of the players into two groups such that there are no EPR-pairs between the groups. Such a state cannot be converted to a $\GHZ^k_2$ state by SLOCC.
\end{proof}

Asymptotically, we can improve the relationship stated in \cref{logrank}, as follows.

\begin{theorem}[Asymptotic upper bound]\label{asympupperbound}
For any $\varepsilon > 0$, there is an $n_0$ such that for all $f: [m]^k \to \{0,1\}$, if $\NQ(f)>n_0$, then
\[
\NQ_0(f) \leq \frac{(k+\varepsilon)}{2} \NQ(f).
\]
\end{theorem}

To prove \cref{asympupperbound} we use the theory of asymptotic SLOCC conversion rates.

\begin{definition} Given tensors $\psi \in V_1\otimes \cdots \otimes V_k$ and $\phi\in W_1\otimes \cdots \otimes W_k$, we say that $\psi$ can be transformed into $\phi$ via SLOCC operations, if there exist linear transformations $A_i : V_i \to K_i$ such that $\phi = (A_1 \otimes \cdots \otimes A_k)\psi$; and we write $\psi \xrightarrow{\mathrm{SLOCC}} \phi$.
Define
\begin{align*}
\omega_n(\psi,\phi) &= \frac{1}{n} \inf\{m \in \NN_{\geq1} \mid \psi^{\otimes m} \xrightarrow{\mathrm{SLOCC}} \phi^{\otimes n}\}\\
\shortintertext{and}
\omega(\psi,\phi) &= \lim_{n\to\infty}\omega_n(\psi,\phi).
\end{align*}
\end{definition}

\begin{lemma}
The limit $\omega(\psi,\phi)$ exists and for all $n$ the inequality $\omega_n(\psi,\phi) \geq \omega(\psi,\phi)$ holds; in other words, $\omega_n = \omega + o(1)$.
\end{lemma}

\begin{theorem}[Vrana-Christandl \cite{vrana}]\label{vranacor} Let $\GHZ_2^{K_k}$ be the $k$-party tensor consisting of EPR-pairs between any parties. Then
\[
\omega(\GHZ_2^{K_k}, \GHZ_2^k) = \frac{1}{k-1}.
\]
In other words, for any $\varepsilon > 0$, there is an $n_0$ such that for all $n>n_0$,
\[
(\GHZ_2^{K_k})^{\otimes n(\frac{1}{k-1}+\varepsilon )}  \xrightarrow{\mathrm{SLOCC}} (\GHZ_2^{k})^{\otimes n}.
\]
\end{theorem}

\begin{proof}[\bf Proof of \cref{asympupperbound}]
Creating $\GHZ_2^{K_k}$ in the $\NQ_0$-model costs $\binom{k}{2}$ messages. Asymptotically, we can transform $1/(k-1)$ copies of $\GHZ_2^{K_k}$ to one copy of $\GHZ_2^k$ by SLOCC. More precisely, by \cref{vranacor}, for any $\varepsilon > 0$, there is an $n_0$ such that for all $n>n_0$,
\[
(\GHZ_2^{K_k})^{\otimes \frac{n}{k-1}+\varepsilon n}  \xrightarrow{\mathrm{SLOCC}} (\GHZ_2^{k})^{\otimes n}.
\]
We conclude that, for any $\varepsilon > 0$, there is an $n_0$ such that for all $n>n_0$, $\binom{k}{2}(\frac{n}{k-1}+\varepsilon n) = ((k+\varepsilon')n)/2$ messages are sufficient to generate $(\GHZ_2^{k})^{\otimes n}$ by SLOCC. 

To prove the theorem, suppose we have an $\NQ$-protocol for $f$ which uses a $\GHZ$ state of rank $2^n$ and no communication. Consider the following $\NQ_0$-protocol for $f$. Create a $\GHZ$-state of rank $2^n$ by sending $\frac{(k+\varepsilon')n}{2}$ messages and then continue with the $\NQ$-protocol.
\end{proof}

The following proposition says that the asymptotic relationship of \cref{asympupperbound} is tight.

\begin{proposition}[$n$-bit inputs]\label{impr}
Let $f:[2^n]^k \to \{0,1\}$ be the function defined by $f(x_1,\ldots,x_k) = [x_1=x_2=\cdots=x_k]$ for $x_i \in [2^n]$. Then we have $\NQ(f) = n$ and $\NQ_0(f) \geq \tfrac{k}{2} \NQ(f)$.
\end{proposition}
\begin{proof}
As in the previous proof, note that the tensor corresponding to $f$ is $\GHZ^k_{2^n}$. Suppose there is an $\NQ_0$ protocol using $r$ messages. View the communication pattern of this protocol as an undirected multigraph $G$ (i.e.~parallel edges are allowed) on $k$ vertices. Note that $G$ has $r$ edges. Let $E=\GHZ_2^G$ be the tensor that has an EPR pair at every edge in $G$. The protocol yields an SLOCC transformation of $E$ to $\GHZ^k_{2^n}$. Let~$\ell$ be the minimal number of edges across any cut of $G$. Then $\ell$ is at most the minimal degree $d$ of~$G$. The sum of all degrees in $G$ equals $2r$, so $k \ell \leq k d \leq 2r$, which implies the inequality $r \geq k \ell/2$. The number $\ell$ is equal to $\min_{S\subseteq[k]}\log_2 \rk_S(E)$, where $\rk_S(E)$ denotes the rank of $E$ after flattening according to the set $S$. This value cannot increase under any SLOCC transformation. Now note that $\log_2 \rk_{\{i\}}(\GHZ^k_{2^n}) = n$ for any $i\in [k]$, so $\ell \geq n$. We conclude that $r \geq kn/2$.
\end{proof}

\begin{remark}
Another way to prove \cref{impr} is to first symmetrize the protocol to obtain an SLOCC transformation of a state $E$ with  $\log_2 \rk_{\{i\}}(E) = (k-1)!2r$ to the state $\GHZ^k_{2^{k!n}}$. We have $\log_2 \rk_{\{i\}}(\GHZ^k_{2^{k!n}}) = k!\,n$. Since $\log_2 \rk_{\{i\}}$ is an SLOCC-monotone, we obtain the inequality $(k-1)!\,2r \geq k!\,n$ and hence $r\geq kn/2$.
\end{remark}

\section{Cyclic equality problem}\label{cep}

The two-player equality problem $\EQ_n$ is the problem of Alice and Bob having to decide whether their $n$-bit inputs are equal. Since the identity matrix has full support rank, we have $\NQ(\EQ_n) = n$. We generalize $\EQ_n$ to multiple players as follows. Let $G$ be an undirected graph. Let $\EQ_n^G$ be the problem of $|G|$ players having to solve the $n$-bit equality problem between players connected by edges. (Note that this definition naturally generalizes to hypergraphs.) If $G$ is a bipartite graph, one easily sees that by grouping the players we can transform the problem into an equality problem on $en$ bits $\EQ_{en}$, where $e$ is the number of edges in the graph. Therefore $\NQ(\EQ_n^G) = en$, that is, the trivial protocol is optimal for bipartite graphs. On the other hand, if $G$ contains an odd cycle, then this argument fails. In the rest of this paper we will focus on the extreme case of $G$ being an odd cycle and investigate the complexity of the corresponding equality problem.

\begin{definition}
The $k$-player \emph{cyclic equality problem} on $n$ bits $\EQ_n^{C_k}$ is the function
\[
\EQ_n^{C_k}: ([2^n]\times[2^n])^k \to \{0,1\}:\, (a_1 b_1, \ldots, a_k b_k) \mapsto 
\begin{cases} 
1 & \textnormal{if $b_1 = a_2,\, b_2 = a_3,\, \ldots,\, b_k = a_1$}\\
0 & \textnormal{otherwise},
\end{cases}
\]
that is, the players are arranged in a circle; player $i$ receives two $n$-bit inputs $a_i, b_i$ and has to decide whether $a_i=b_{i-1}$ and $b_i = a_{i+1}$, where the indices are taken modulo $k$.
\end{definition}

It turns out that the tensor corresponding to this function is a generalisation of the \emph{matrix multiplication tensor}, one of the central objects of study in algebraic complexity theory. This tensor arises as follows in algebraic complexity theory. Consider the bilinear map
\[
\CC^{m\times m} \times \CC^{m\times m} \to \CC^{m\times m} : (A, B) \mapsto AB
\]
which multiplies two complex $m\times m$ matrices. Any bilinear map $U\times V\to W$ corresponds canonically to a tensor in $U\otimes V\otimes W$. The number of multiplications in the field $\CC$ necessary to perform the bilinear map is equal to the tensor rank of the corresponding tensor, up to a factor 2. The tensor corresponding to the matrix multiplication map is
\[
\langle m,m,m \rangle \coloneqq \sum_{x\in [m]^3} \ket{x_1 x_2}\!\ket{x_2 x_3}\!\ket{x_3 x_1}.
\]
A natural generalisation of the tensor $\langle m,m,m \rangle$ to a $k$-party tensor is the so-called \emph{iterated matrix multiplication} tensor
\[
\IMM_m^k \coloneqq \sum_{x\in [m]^k} \ket{x_1 x_2}\!\ket{x_2 x_3}\cdots\ket{x_k x_1}.
\]
Clearly, $\IMM_m^3 = \langle m,m,m \rangle$. The tensor $\IMM_m^k$ corresponds to the multilinear map
\[
(\CC^{m\times m})^{\times k} \to \CC : (A_1,A_2,\ldots, A_k) \mapsto \tr (A_1 A_2\cdots A_k)
\]
which computes the trace of the product of $k$ matrices. We note that, when viewed as a polynomial in the matrix entries, $\IMM_m^k$ plays a special role in the field of arithmetic circuits and geometric complexity theory. Namely, $\IMM^k_3$ is complete for the class $\VPe$ of families of polynomials computable by small formulas \cite{ben1992computing}, and $\IMM^k_k$ is complete for the class $\VQP$, for which the determinant is also complete \cite{blaser2001complete}. The following connection between iterated matrix multiplication and cyclic equality is readily observed.

\begin{proposition}\label{cycmamu}
The tensor corresponding to the cyclic equality function $\EQ_n^{C_k}$ on $n$ bits is the iterated matrix multiplication tensor $\IMM_{2^n}^k$ with $2^n \times 2^n$ matrices. Therefore, we have the equalities $\NQ(\EQ_n^{C_k}) = \log_2 \supprank(\IMM_{2^n}^k)$ and $\borderNQ(\EQ_n^{C_k}) = \log_2 \bsupprank(\IMM_{2^n}^k)$
\end{proposition}

The remainder of this paper is organized as follows. In the following four paragraphs we do the following: (1) we show that  in the classical model, the naïve protocol in which every player broadcasts his input is optimal; (2) we show that when $k$ is even the naïve protocol is optimal quantumly; (3) we exhibit nontrivial protocols when $n=1$ and $k=3$ or $k=5$; (4) we show nontrivial lower bounds on the quantum complexity by use of Young flattenings. Finally, in the last section, we show that the Strassen laser method yields nontrivial protocols for all odd $k\geq 3$, asymptotically.

\paragraph{Classical lower bound with the fooling set method.} 
We will show that in the classical situation the trivial protocol is always optimal. To prove a lower bound on the classical complexity of the cyclic equality problem we use the fooling set method.

\begin{theorem} The classical nondeterministic communication complexity $\N(\EQ_n^{C_k})$ of the cyclic equality problem equals $kn$.
\end{theorem}
\begin{proof}
Let $S\subseteq[2^{2n}]^k$ be the set of 1-inputs of the function $\EQ_n^{C_k}$. This set has size~$2^{kn}$. Let $\Pi$ be a classical protocol for $\EQ_n^{C_k}$ and denote by $\Pi_r(x_1,\ldots,x_k)$ the sequence of messages sent by the players in the protocol $\Pi$ on input $x\in [2^{2n}]^k$ and private randomness~$r\in[m]^k$. Suppose there are distinct 1-inputs $x,y \in S$ and private randomnesses $r,s \in [m]^k$ such that $\Pi_r(x_1,\ldots,x_k) = \Pi_s(y_1,\ldots,y_k)$. There is an $i$ such that $x_i \neq y_i$, say $i=1$. We have $\Pi_r(x_1,\ldots,x_k) = \Pi_{(r_1,s_2,\ldots,s_k)}(x_1,y_2,\ldots,y_k)$, so the protocol outputs 1 on input $x_1,y_2,\ldots,y_k$ with randomness $(r_1,s_2,\ldots,s_k)$. However, $x_1,y_2,\ldots,y_k$ is a 0-input, a contradiction. Therefore, $\Pi_r(x_1,\ldots,x_k) \neq \Pi_s(y_1,\ldots,y_k)$. We conclude that $\N(\EQ_n^{C_k}) \geq \log_2(|S|)$.
\end{proof}

\paragraph{An even number of quantum players.} When $k$ is even, the cycle graph $C_k$ is bipartite, and, as mentioned above, the best protocol for an equality problem on a bipartite graph is the trivial protocol. We record this statement in terms of border support rank in the following proposition.

\begin{proposition}\label{evencase} For even $k$,\, $m^k \leq \bsupprank(\IMM_{m}^k)$. As a consequence, we have the equalities $\borderNQ(\EQ_n^{C_k}) = \NQ(\EQ_n^{C_k}) = kn$.
\end{proposition}
\begin{proof}
Let $t$ be a tensor with the same support as $\IMM_m^k \in (\CC^{m^2})^{\otimes k}$.
Label the players with the numbers $1,2,\ldots,k$. Group the \emph{even} players together and group the \emph{odd} players together and flatten the tensor $t$ accordingly into a matrix $A$ in $(\CC^{m^2})^{\otimes{k/2}}\otimes (\CC^{m^2})^{\otimes{k/2}}$. The matrix $A$ has the same support as the identity matrix in $(\CC^{m^2})^{\otimes{k/2}}\otimes (\CC^{m^2})^{\otimes{k/2}}$ and thus has rank $m^k$.
\end{proof}

Note that for odd $k$ the above proof yields the lower bound $m^{k-1} \leq \bsupprank(\IMM_{m}^k)$. We will show in \cref{Yfl} that this lower bound is not tight.

\paragraph{Nontrivial 3-player and 5-player quantum protocols.} In the 3-player situation, Strassen's celebrated decomposition of the tensor $\IMM_2^3=\langle 2,2,2\rangle$ into a sum of 7 simple tensors \cite{strassen1969gaussian} gives a nontrivial protocol for $\EQ_1^{C_3}$, and thus $\NQ(\EQ_1^{C_3}) \leq \log_2(7)$. 
We show that for 5 players there also exists a nontrivial protocol for $\EQ_1^{C_5}$, as follows. Recall that we have defined $\IMM_2^5 = \sum_{i\in [2]^5} \ket{i_1 i_2}\!\ket{i_2 i_3}\!\ket{i_3 i_4}\!\ket{i_4 i_5}\!\ket{i_5 i_1}$. Observe that an upper bound $\rank(\IMM_2^5) \leq r$ implies $\rank(\IMM_n^5) \leq \Oh(n^{\log_2(r)})$ by taking tensor powers of $\IMM_2^5$.

\begin{theorem}\label{smallprot}
$\rank(\IMM_2^5) \leq 31$, and thus $\NQ(\EQ_1^{C_5}) \leq \log_2(31)$.
\end{theorem}
\begin{proof}
Let $\ket{\boldsymbol{-}}\coloneqq \ket{1}-\ket{2}$, $\ket{\boldsymbol{+}} \coloneqq \ket{1} + \ket{2}$ and $\ket{\Phi^+} = \ket{11} + \ket{22}$. Let $\Cyc_{5} \coloneqq \sum_{\sigma\in C_5} \sigma$ be the cyclic symmetrizer acting on $(\CC^{4})^{\otimes 5}$ by permuting the 5 parties, and moreover let $\Sym_2 \coloneqq \sum_{\sigma\in S_2} \sigma$ be a `local symmetrizer' acting diagonally on $(\CC^2)^{\otimes 10}$ by permuting the basis states $\ket{1}$ and $\ket{2}$ of each $\CC^2$. Let
\begin{alignat*}{3}
t \coloneqq & \,-\, \ket{\boldsymbol{-}1}\ket{11}\ket{11}&\ket{1\boldsymbol{+}}\ket{22} &\\[0.1em]
& \,-\, \ket{\boldsymbol{-}1}\ket{12}\ket{21}&\ket{1\boldsymbol{+}}\ket{22} &\\[0.1em]
& \,-\, \ket{\Phi^+} \ket{22} \ket{\boldsymbol{-}1}&\ket{1\boldsymbol{+}}\ket{22}&.
\end{alignat*}
By direct computation, we see that $\IMM_2^5 = \Cyc_5\bigl(\Sym_2(t)\bigr) + \ket{\Phi^+}^{\otimes 5}$. We observe that the right hand side yields a sum of 31 simple tensors.
\end{proof}

We have a proof generalizing \cref{smallprot} to $\rank(\IMM_2^k) \leq 2^k - 1$ for all odd $k$, which will appear in a forthcoming paper.

\paragraph{Quantum lower bound with Young flattenings.} Let $t\in V_1\otimes V_2\otimes V_3$ be some 3-tensor. By grouping $V_1$ and $V_2$, the tensor $t$ can be viewed as a matrix $A\in (V_1\otimes V_2) \otimes V_3$; this is called a \emph{flattening}. The rank of the flattening $A$ is a lower bound for the border rank of~$t$ and thus we obtain lower bounds on the border rank of tensors by computing the rank of their flattenings. However, this type of lower bound can never be bigger than the dimension of any local space $V_i$, and there do exist tensors with border rank larger than the local dimensions, for example the matrix multiplication tensor $\langle 2,2,2\rangle$.

One approach to overcome this `local dimension limitation' is as follows. We let ${\phi : V_2 \to W_1 \otimes W_2}$ be a linear map such that $\rank(\phi(v)) \leq e$ for all $v\in V_2$. By applying~$\phi$ to the central tensor leg of $t$, we transform $t$ into a 4-tensor $s\in V_1\otimes W_1 \otimes W_2 \otimes V_3$. Next, we flatten $s$ to a matrix $A\in (V_1 \otimes W_1) \otimes (W_2 \otimes V_3)$. The rank of $A$ divided by $e$ is a lower bound for the border rank of $t$. We will be using a specific linear map $\phi$ which originates from the representation theory of the general linear group. When one takes such representation theoretic maps $\phi$ to construct a flattening as above one speaks of a \emph{Young flattening} \cite{landsberg2011new}. An early appearance of this type of flattening can be recognized in the work of Strassen~\cite{strassen1983rank}. The following lower bound is obtained with a Young flattening.

\begin{theorem}\label{Yfl} For odd $k\geq 3$, \,
$(2n^2-n)n^{k-3} \leq \bsupprank(\IMM_n^k)$. As a consequence, we have the lower bound $(k-1)n + \log_2(2-\tfrac{1}{n}) \leq \borderNQ(\EQ_n^{C_k})$.
\end{theorem}
\begin{proof}
Let $k=3$. The proof for odd $k>3$ goes similarly after having grouped the $k$ parties appropriately to 3 parties. For a vector space $V$, let $\wedge^a V$ be the $a$th exterior power of $V$. Define the linear map
\begin{alignat*}{1}
\phi \,\colon &\CC^{2n-1} \,\to\,  \wedge^p \CC^{2n-1} \otimes \wedge^{p+1} \CC^{2n-1} \\
	   &\!\ket{j} \,\mapsto\, \!\!\!\sum_{j_1<\cdots<j_p}\!\!\!\!  \ket{j_1}\!\wedge\! \cdots \!\wedge\! \ket{j_p} \,\otimes\, \ket{j_1}\!\wedge\! \cdots \!\wedge\! \ket{j_p}\! \wedge\! \ket{j},
\end{alignat*}
and note that the rank of the matrix $\phi(v)$ equals $\binom{2n-2}{p}$ for any $v\in \CC^{2n-1}$. We consider the tensor
\[
t_1 \coloneqq \sum_{i} \alpha_{i_1, i_2, i_3}\ket{i_1 i_2}\ket{i_2 i_3} \ket{i_3 i_1}\, \in\, \CC^{n^2} \otimes \CC^{n^2} \otimes \CC^{n^2}\!,
\]
where $i$ runs over $[n]^3$ and the $\alpha_{i_1, i_2, i_3}$ are nonzero complex numbers. The border rank of $t_1$ is at least the border rank of
\[
t_2 \coloneqq \sum_{i} \alpha_{i_1, i_2, i_3}\ket{i_1 i_2}\ket{i_2 + i_3 - 1} \ket{i_3 i_1}\, \in\, \CC^{n^2} \otimes \CC^{2n-1} \otimes \CC^{n^2}\!.
\]
Apply $\phi$ to the central tensor leg of $t_2$ and then flatten to obtain
\[
A \coloneqq  \sum_{i} \sum_{j_1<\cdots<j_p} \alpha_{i_1, i_2, i_3} \ket{i_1 i_2}  \ket{j_1}\! \wedge\! \cdots \!\wedge\! \ket{j_p} \, \otimes\, \ket{j_1} \!\wedge\! \cdots \!\wedge\! \ket{j_p}\! \wedge\! \ket{i_2 + i_3 - 1} \ket{i_3 i_1}.
\]
View $A$ as a direct sum of $n$ matrices $A_{i_1} \in (\CC^n \otimes \wedge^p \CC^{2n-1}) \otimes (\wedge^{p+1} \CC^{2n-1} \otimes \CC^n)$; the matrix $A_{i_1}$ corresponds to the linear map
\[
f_{i_1} \coloneqq  \ket{i_2}  \ket{j_1}\! \wedge\! \cdots \!\wedge\! \ket{j_p} \mapsto \sum_{i_3}  \alpha_{i_1, i_2, i_3} \ket{j_1} \!\wedge\! \cdots \!\wedge\! \ket{j_p}\! \wedge\! \ket{i_2 + i_3 - 1} \ket{i_3}.
\]
Let $p=n-1$. We claim that every matrix $A_{i_1}$ is upper triangular with elements $\alpha_{i_1,i_2,i_3}$ on the diagonal, up to permutations of the rows and columns. Assuming the claim is true, we get that $\rank(A) = \sum_{i_1}\rank(A_{i_1}) = n \dim(\CC^n\otimes \wedge^{n-1} \CC^{2n-1}) = n^2 \binom{2n-1}{n-1}$.  This implies the lower bound $\bsupprank(\IMM_n^3) \geq n^2 \binom{2n-1}{n-1} / \binom{2n-2}{n-1} = 2n^2-n$.

To prove this claim we define a partial order on the basis elements $\ket{j_1}\wedge\cdots \wedge\ket{j_n}\otimes \ket{\ell}$ of the target space of $A_{i_1}$. We will use the same partial order as Landsberg and Michałek~\cite{landsberg2016geometry}. Denote the basis elements of the target space by $(P,\ell)$ with $P$ an $n$-subset of~$[2n-1]$ and $\ell \in [n]$. Let $(P_1,\ell_1)$ and $(P_2, \ell_2)$ be two such basis elements and define $\ell\coloneqq \min(\ell_1,\ell_2)$. We say $(P_1,\ell_1) < (P_2, \ell_2)$
\begin{enumerate}
\item if the ordered sequence of the $\ell$ smallest elements in $P_2$ is lexicographically smaller than the ordered sequence of the $\ell$ smallest elements in $P_1$;
\item or if the sequences of $\ell$ smallest elements are equal and $\ell_1 < \ell_2$.
\end{enumerate}
One checks that this defines a partial order and that the unique minimal element in this order is $(\{n,\ldots,2n-1\},1)$. For example, with $n=2$ the partial order has the following Hasse diagram.
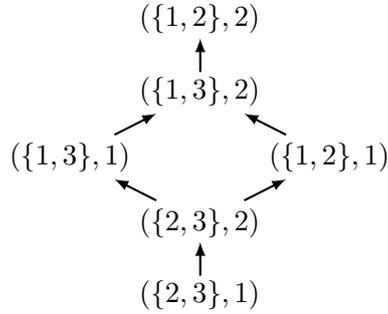
\begin{center}
\begin{tikzpicture}[node distance=0.8em, >=latex]
    \node (1) at (0,0) {$(\{1,2\},2)$};
    \node [below = of 1] (2)  {$(\{1,3\},2)$};
    \node [below right=of 2, xshift=-1em] (3) {$(\{1,2\},1)$};
	\node [below left =of 2, xshift=1em] (4) {$(\{1,3\},1)$};
	\node [below right=of 4, xshift=-1em] (5) {$(\{2,3\},2)$};
	\node [below=of 5] (6) {$(\{2,3\},1)$};

    \draw [thick, shorten <=-2pt, shorten >=-2pt, ->] (2) -> (1);
    \draw [thick, shorten <=-2pt, shorten >=-2pt, ->] (3) -> (2);
    \draw [thick, shorten <=-2pt, shorten >=-2pt, ->] (4) -> (2);
    \draw [thick, shorten <=-2pt, shorten >=-2pt, ->] (5) -> (4);
    \draw [thick, shorten <=-2pt, shorten >=-2pt, ->] (5) -> (3);
    \draw [thick, shorten <=-2pt, shorten >=-2pt, ->] (6) -> (5);
\end{tikzpicture}
\end{center}

We prove the claim by induction on $<$, with the minimal element as a base case. For now let all the $\alpha_{i_1,i_2,i_3}$ be 1. First, under $A_{i_1}$ we have
\[
\ket{n} \otimes \ket{n+1}\wedge \cdots \wedge \ket{2n-1}\mapsto \ket{n}\wedge\cdots\wedge\ket{2n-1} \otimes \ket{1},
\]
so the minimal element $(\{n,\ldots,2n-1\},1)$ is in the image of $A_{i_1}$. Let $(P,\ell)$ be in the target space of $A_{i_1}$ and assume that every $(P',\ell')$ with $(P',\ell')< (P, \ell)$ is in the image. Write $P = (p_1, \ldots, p_n)$ with $p_1\leq\cdots\leq p_n$. Under $A_{i_1}$ we have,
\[
\ket{p_1}\wedge \cdots \widehat{p_\ell} \cdots \wedge \ket{p_n} \otimes \ket{1+p_\ell - \ell} \mapsto \sum_m \ket{p_1}\wedge\cdots\widehat{p_\ell}\cdots\wedge\ket{p_n}\wedge \ket{p_\ell - \ell + m} \otimes \ket{m}.
\]
Taking $m=\ell$, one sees that the basis element $(P,\ell)$ is present in the sum. Moreover, for any other $(P',m)$ appearing in the sum we have $(P',m) < (P,\ell)$. Indeed, if $m>\ell$, then $p_\ell - \ell + m > p_\ell$, so the smallest $\ell$ elements in $P'$ are lexicographically larger than the smallest $\ell$ elements in $P$, meaning $(P',m) < (P,\ell)$ by rule 1; if $m < \ell$, then $p_m\leq p_\ell-\ell+m < p_\ell$, so the $m$ smallest elements of $P'$ and $P$ are equal, meaning $(P',m) < (P,\ell)$ by rule 2.  Therefore, the basis element $(P,\ell)$ is in the image. This argument shows that $A_{i_1}$ has full rank. Moreover, this argument shows that, up to a permutation of the rows and columns, the matrix $A_{i_1}$ is upper triangular with ones on the diagonal. Repeating this argument with general values for  $\alpha_{i_1,i_2,i_3}$ proves the claim.
\end{proof}

\begin{remark}
The lower bound in \cref{Yfl} improves a lower bound of Ikenmeyer on the border support rank of $\IMM_n^3$ \cite[8.2.17]{ikenmeyer2013geometric}.
\end{remark}

\section{Strassen's laser method for iterated matrix multiplication}\label{lasersec} In this section we show that $\NQ(\EQ_n^{C_k}) < kn$ for odd $k$. We will prove this result in the language of algebraic complexity theory.

\begin{definition}
Define $\omega_k \coloneqq \inf \{ \alpha \in \RR \mid \rank( \IMM_n^k ) \in \Oh(n^\alpha)\}$. We call this the \emph{exponent} of iterated matrix multiplication. Define $\omega_{\mathrm{s},k} \coloneqq \inf \{ \alpha \in \RR \mid \supprank( \IMM_n^k ) \in \Oh(n^\alpha)\}$. We call this the \emph{support rank exponent} of iterated matrix multiplication.
\end{definition}

Asymptotically, we have $\NQ(\EQ_n^{C_k}) \leq \omega_{\mathrm{s},k}\, n + \Oh(1) \leq \omega_k\, n + \Oh(1)$. The exponents $\omega_3$ and $\omega_{\mathrm{s},3}$ are known as $\omega$ and $\omega_{\mathrm{s}}$ in the literature. The support rank exponent of matrix multiplication was first studied by Cohn and Umans \cite{cohn2013fast}. The best upper bound on~$\omega_s$ comes from the upper bound $\omega \leq 2.3728639$ of Le~Gall \cite{le2014powers}. Interestingly, Cohn and Umans show the relationship
\[
\omega \leq (3 \omega_{\mathrm{s}} - 2)/2.
\]
Therefore, one way of finding upper bounds on $\omega$ is to construct an efficient quantum communication protocol for the cyclic equality problem $\EQ_n^3$. 

For any $k$ we have $k-1 \leq \omega_k \leq k$, and if $k$ is even, then $\omega_k = k$ (\cref{evencase}). The aim of this section will be to show: if $k\geq 3$ is odd, then
\[
\omega_k < k.
\]

\paragraph{Schönhage $\tau$-theorem.}
In this section we will generalize some tools for obtaining upper bounds on the exponent of $\omega_3$ to all exponents $\omega_k$, in particular, we generalize the Schönhage $\tau$-theorem. The proofs in this section are straightforward generalizations of the proofs for $k=3$ which can be found in \cite{blaser2013fast}. In the next paragraph, we will use Strassen's laser method to show that $\omega_k < k$ for all odd $k$.

First we recall an important relationship between border rank and rank. We use the following more precise notion of border rank. Let $h\in \NN$ and let $t$ be a tensor in $\CC^{\otimes m_1} \otimes \cdots \otimes \CC^{\otimes m_k}$. Define $\rank_h(t)$ to be the minimum number $r$ such that there exist vectors $v_i^j \in (\CC[\varepsilon])^{m_j}$ that satisfy $\sum_{i=1}^r v_i^1\otimes \cdots\otimes v_i^k = \varepsilon^h t + \Oh(\varepsilon^{h+1})$. A well-known but nontrivial result is that $\borderrank(t) = \min_h \rank_h(t)$. It is not hard to show that $\rank_{h+h'}(t\otimes t') \leq \rank_h(t) \rank_{h'}(t')$. The relationship we are talking about is the following.

\begin{proposition}
For every $h,k \in \NN$, there is a number $c_h$ such that for all tensors $t \in \CC^{m_1}\otimes \cdots \otimes \CC^{m_k}$, $\rank(t) \leq c_h \rank_h(t)$. The number $c_h$ depends polynomially on $h$.
\end{proposition}
\begin{proof}
Let $t$ be a tensor in $\CC^{m_1}\otimes \cdots \otimes \CC^{m_k}$ with $\rank_h(t) = r$. Then there are $v_i^j \in (\CC[\varepsilon])^{m_j}$ such that
\[
\sum_{i=1}^r v_i^1 \otimes \cdots \otimes v_i^k = \varepsilon^ht + \Oh(\varepsilon^{h+1}).
\]
Decomposing every $v_i^j$ into $\varepsilon$-homogeneous components $v_i^j = \sum_{a_j=0}^h \varepsilon^{a_j} v_i^j(a_j)$, and collecting powers of $\varepsilon$ gives
\[
\sum_{i=1}^r \sum_{a_1,\ldots,a_k \in [h]} \varepsilon^{a_1+\cdots+a_k}\, v_i^1(a_1) \otimes \cdots \otimes v_i^k(a_k) = \varepsilon^h t + \Oh(\varepsilon^{h+1}).
\]
Taking only the summands such that $a_1 + \cdots + a_k = h$ gives a rank decomposition of $t$. There are $\binom{h+k-1}{k-1}r$ such summands.
\end{proof}

Next, we show that an upper bound on the border rank of `unbalanced' iterated matrix multiplication tensors yields and upper bound on $\omega_k$.

\begin{proposition}\label{symm}
If $\borderrank(\langle n_1,n_2,\ldots,n_k \rangle ) \leq r$, then $\omega_k \leq k \log_{n_1\cdots n_k} r$.
\end{proposition}
\begin{proof}
Let $N = n_1\cdots n_k$. There is an $h$ such that $\rank_h(\langle n_1,\ldots, n_k\rangle) \leq r$. By taking the tensor product of all cyclic shifts of $\langle n_1,\ldots,n_k\rangle$, we get $\rank_{kh}(\langle N,\ldots, N\rangle) \leq r^k$ and thus $\rank_{khs}(\langle N^s,\ldots,N^s\rangle) \leq r^{ks}$ for all $s$. Hence $\rank(\langle N^s,\ldots, N^s\rangle) \leq c_{khs} r^{ks}$ for some number~$c_{khs}$ which is constant in $N$. Therefore,
\[
\omega \leq \log_{N^s}(c_{khs} r^{ks}) = ks \log_{N^s}(r) + \log_{N^s}(c_{khs}).
\]
If $s$ goes to infinity then $\log_{N^s}(c_{khs})$ goes to zero, so  $\omega_k \leq k \log_N(r)$.
\end{proof}

The real workhorse is the following straightforward generalization of a theorem of Schönhage \cite{schonhage1981partial}.

\begin{proposition}[$k$-party Schönhage $\tau$-theorem]\label{tau} Suppose that $r>p$ and 
\[
\borderrank\Bigl(\bigoplus_{i=1}^p \langle n_1^i,n_2^i,\ldots,n_k^i\rangle\Bigr) \leq r.
\]
Define $\tau$ by $\sum_{i=1}^p \bigl(\prod_{j=1}^k n_j^i\bigr)^\tau = r$. Then $\omega_k \leq k\tau$
\end{proposition}

We follow the proof of \cite{blaser2013fast}. We first prove two lemmas. For tensors $s,t\in\CC^{m_1}\otimes \cdots \otimes \CC^{m_k}$, let $s\leq t$ denote the existence of an SLOCC transformation mapping $t$ to $s$. Let $a,b \in \NN+1$.

\begin{lemma}\label{indlem}
Let $t$ be a tensor such that $\rank( t^{\oplus a} ) \leq b$. Then for all $s$, $\rank((t^{\otimes s})^{\oplus a}) \leq \ceil{b/a}^s a$.
\end{lemma}
\begin{proof}
We prove the lemma by induction over $s$. The base case $s=1$ follows from the assumption. For the induction step, we have
\[
(t^{\otimes s+1})^{\oplus a} = t^{\oplus a} \otimes t^{\otimes s} \leq \GHZ_b \otimes t^{\otimes s} = (t^{\otimes s})^{\oplus b},
\]
and thus, by the induction hypothesis,
\[
\rank((t^{\otimes s+1})^{\oplus a}) \leq \rank( (t^{\otimes s})^{\oplus b}) \leq \rank((t^{\otimes s})^{\oplus (\ceil {b/a} a)}) \leq \ceil{\tfrac{b}{a}} \ceil{\tfrac{b}{a}}^s a \leq \ceil{\tfrac{b}{a}}^{s+1} a,
\]
proving the lemma.
\end{proof}

\begin{lemma}\label{corind}
If $\rank(\langle n_1,n_2,\ldots,n_k \rangle ^{\oplus a}) \leq b$, then $\omega \leq k \log_{n_1\cdots n_k} \ceil{b/a}$.
\end{lemma}
\begin{proof}
The inequality $\rank(\langle n_1,n_2,\ldots,n_k \rangle ^{\oplus a}) \leq b$ implies by \cref{indlem} the inequality $\rank(\langle n_1^s,n_2^s,\ldots,n_k^s\rangle) \leq \rank(\langle n_1^s,n_2^s,\ldots,n_k^s\rangle^{\oplus a}) \leq \ceil{b/a}^s a$ which by \cref{symm} yields
\[
\omega_k \leq k\frac{s \log{\ceil{\tfrac{b}{a}}} + \log(a)}{s \log(n_1\cdots n_k)},
\]
which goes to $k \log{\ceil{b/a}} /  \log(n_1\cdots n_k)$ when $s$ goes to infinity.
\end{proof}

\begin{proof}[\bf Proof of \cref{tau}]
We assume $\borderrank\bigl(\bigoplus_{i=1}^p \langle n_1^i,n_2^i,\ldots,n_k^i\rangle\bigr) \leq r$. This implies that there is an $h\in \NN$ such that $\rank_h\bigl(\bigoplus_{i=1}^p \langle n_1^i,n_2^i,\ldots,n_k^i \rangle\bigr) \leq r$. Taking the $s$th tensor power gives $\rank_{hs}\bigl((\bigoplus_{i=1}^p \langle n_1^i,n_2^i,\ldots,n_k^i \rangle)^{\otimes s}\bigr) \leq r^s$. We expand the tensor power to get
\[
\rank_{hs}\Bigl( \bigoplus_{\sigma} \Bigl(\bigotimes_{i=1}^p \bigl\langle (n_1^i)^{\sigma_i}, (n_2^i)^{\sigma_i}, \ldots, (n_k^i)^{\sigma_i}  \bigr\rangle \Bigr)^{\oplus \binom{s}{\sigma_1,\ldots,\sigma_p}} \Bigr) \leq r^s,
\]
where the first direct sum is over all $p$-tuples $\sigma$ of nonnegative integers with sum $s$. We can also write this inequality as
\[
\rank_{hs}\Bigl( \bigoplus_{\sigma} \bigl\langle \textstyle\prod_i (n_1^i)^{\sigma_i}, \ldots, \textstyle\prod_i (n_k^i)^{\sigma_i}  \bigr\rangle \Bigr)^{\oplus \binom{s}{\sigma_1,\ldots,\sigma_p}} \Bigr) \leq r^s.
\]
There exists a number $c_{hs}$ depending polynomially on $h$ and $s$ such that
\[
\rank\Bigl( \bigoplus_{\sigma}\, \bigl\langle \textstyle\prod_i (n_1^i)^{\sigma_i}, \ldots, \prod_i (n_k^i)^{\sigma_i}  \bigr\rangle ^{\oplus \binom{s}{\sigma_1,\ldots,\sigma_p}} \Bigr) \leq c_{hs}\, r^s.
\]
Define $\tau$ by $\sum_{i=1}^p \bigl(\prod_{j=1}^k n_j^i\bigr)^\tau = r$. Then $\sum_{\sigma} \binom{s}{\sigma_1,\ldots,\sigma_p} \bigl(\prod_i(n_1^i)^{\sigma_i} \cdots \prod_i(n_k^i)^{\sigma_i}\bigr)^{\tau} = r^s$. In this sum, consider the maximum summand and fix the corresponding $\sigma$. Define the numbers $n_j \coloneqq \prod_i(n_j^i)^{\sigma_j}$. Let $a \coloneqq \binom{s}{\sigma_1,\ldots,\sigma_p}$ and $b \coloneqq r^s c_{hs}$. We apply \cref{corind} to the inequality $\rank(\langle n_1',\ldots,n_k'\rangle^{\oplus a}) \leq b$ to obtain
\[
\omega_k \leq k \tau + \frac{(p-1)\log(s+1) + \log(c_{hs})}{\log(n_1\cdots n_k)},
\]
which goes to $k \tau$ when $s$ goes to infinity. (See \cite{blaser2013fast} for more details.)
\end{proof}

\paragraph{Strassen's laser method.} We will now use Strassen's laser method to prove the main result of this section.
\begin{theorem}\label{nontrivomega}
For any odd $k$ we have $\omega_k < k$.
\end{theorem}

We will give a proof for the case $k=5$, the other cases being similar. Define the $5$-tensor $\Str_q^5 = \sum_{i=1}^q \ket{ii000} + \ket{0ii00}$ in $\CC^{q+1} \otimes \CC^q \otimes \CC^{q+1} \otimes  \CC \otimes \CC$.
\begin{proposition}
$\borderrank(\Str^5_q) \leq q+1$.
\end{proposition}
\begin{proof} Expanding $\sum_{i=1}^q (\ket{0} + \varepsilon\ket{i})\ket{i}(\ket{0} + \varepsilon\ket{i})\ket{0}\!\ket{0}$ gives 
\[
\sum_{i=1}^q \ket{0i000} + \varepsilon \sum_{i=1}^q \ket{ii000}+\ket{0ii00} + \Oh(\varepsilon^2).
\]
Subtracting $\ket{0}\!\bigl(\sum_{i=1}^q\ket{i}\bigr)\! \ket{000}$ yields $\varepsilon\,\Str_q^5 + \Oh(\varepsilon^2)$.
\end{proof}

Define the tensor $\langle n_1,n_2,n_3,n_4,n_5 \rangle$ to be 
\[
\sum_{x\in[n_1]\times\cdots \times[n_5]} \ket{x_1x_2}\ket{x_2x_3}\ket{x_3x_4}\ket{x_4x_5}\ket{x_5x_1}.
\]
So $\IMM_n^5 = \langle n,n,n,n,n \rangle$.

\begin{proposition}\label{subrank}
$\GHZ_2^{5} \leq \langle 2,2,2,2,2 \rangle$.
\end{proposition}
\begin{proof}
Let $\phi$ be the map $\ket{ab} \mapsto \delta_{[a=b]}\ket{a}$. Apply $\phi^{\otimes 5}$ to $\langle 2,2,2,2,2 \rangle$. This yields one copy of $\GHZ_2^{[5]}$.
\end{proof}
\begin{remark}
We mention that the subrank result of \cref{subrank} can by improved asymptotically in the sense that rate $\omega(\langle 2,2,2,2,2 \rangle, \GHZ^{5}) = 1/2$ \cite{vrana}. Using this fact in the proof of \cref{nontrivomega} gives the slightly better upper bound $\omega_k \leq \log_q((q+1)^k/4)$.
\end{remark}

For the proof of \cref{nontrivomega} we have to define the notion of the decomposition of the support of a tensor and the corresponding inner and outer structure of a tensor. Let $I_1,\ldots,I_k$ be finite sets. A \emph{decomposition} $\mathcal{D}$ of $I_1\times \cdots \times I_k$ is a collection of sets $I_i^j$ such that
\[
I_i = \bigsqcup_{j} I_i^j,
\]
meaning that for every $i$, $\cap_j I_i^j = \emptyset$ and $\cup_j I_i^j = I_i$. Let $t$ be a tensor in $\CC^{m_1} \otimes \cdots \otimes \CC^{m_k}$ and index the basis elements in this space by elements of $[m_1]\times \cdots \times [m_k]$. Let $\mathcal{D}$ be a decomposition of $[m_1]\times \cdots \times [m_k]$. We view $\mathcal{D}$ as a `cut' of $[m_1]\times \cdots \times [m_k]$ into smaller product sets and thus as a `cut' of $t$ into smaller tensors. We define $t|_{I_1^{j_1},I_2^{j_2},\ldots,I_k^{j_k}}$ to be the restriction of $t$ to the basis elements in $I_1^{j_1}\times I_2^{j_2}\times\ldots\times I_k^{j_k}$. These smaller tensors we think of as the `inner structure' of $t$.  We define the `outer structure' of $t$ with respect to $\mathcal{D}$ to be the tensor $t_{\mathcal{D}}$ indexed by sequences $(j_1,\ldots,j_k)$ such that $t_{\mathcal{D}}$ has a 1 at position $(j_1,\ldots,j_k)$ if $t$ restricted to $I_1^{j_1} \times \cdots \times I_k^{j_k}$ is not the zero tensor, and a 0 otherwise.

\begin{proof}[\bf Proof of \cref{nontrivomega}]
We will give a proof for the case $k=5$, the other cases being similar.
Define a block decomposition $\mathcal{D}$ of the support $I_1\times\cdots \times I_5$ of $\Str_q^5$ by
\begin{align*}
I_1 &= \{0\} \cup \{1,\ldots, q\}\\
I_2 &= \{1,\ldots, q\}\\
I_3 &= \{0\} \cup \{1,\ldots, q\}\\
I_4 &= \{0\}\\
I_5 &= \{0\}.
\end{align*}
We have the outer structure $(\Str_q^5)_{\mathcal{D}} = \ket{11000} + \ket{01100}\cong \ket{10100}+\ket{00000}$. Note that this is just an EPR pair between party 1 and 3. The inner structures are $\sum_{i=1}^q \ket{ii000}$ and $\sum_{i=1}^q \ket{0ii00}$, which are also known as $\langle 1,q,1,1,1 \rangle$ and $\langle 1,1,q,1,1 \rangle$. Let $\Cyc_5$ be the map $t\mapsto t \otimes \sigma t \otimes \sigma^2 t \otimes \sigma^3 t\otimes \sigma^4 t$ with $\sigma = (12345)$.
Let $\hat{\mathcal{D}} = \Cyc_5 \mathcal{D}$ be the naturally corresponding decomposition. Then
\begin{equation}\label{master}
\langle 2,2,2,2,2 \rangle^{\otimes s}=(\Cyc_5 \Str_q^5)_{\hat{\mathcal{D}}^{\otimes s}}^{\otimes s}\quad\textnormal{and} \quad \borderrank\bigl((\Cyc_5 \Str_q^5)^{\otimes s}\bigr) \leq (q+1)^{5s}.
\end{equation}
Note how the first statement relies on 5 being odd.

The inner structure of $(\Cyc_5 \Str_q^5)_{\hat{\mathcal{D}}^{\otimes s}}^{\otimes s}$ consists of tensors from $I\coloneqq \{\langle n_1,n_2,n_3,n_4,n_5\rangle \mid n_1\cdots n_5 = q^{5s}\}$.
 Combining equation~\eqref{master} with \cref{subrank} gives that there are $2^s$ elements $t_1,t_2,\ldots \in I$ such that
\[
\borderrank( t_1 \oplus t_2 \oplus \cdots ) \leq (q+1)^{5s}.
\]
Now the $\tau$-theorem says that if we define $\tau$ by
\[
2^{s} (q^{5s})^\tau  = (q+1)^{5s}
\]
then $\omega_5 \leq 5 \tau$. Therefore,
\[
\omega_5 \leq 5\tau \leq \log_q \frac{(q+1)^{5}}{2}
\]
which gives $\omega_5\leq4.84438$. In general, one gets $\omega_k \leq \log_q \frac{(q+1)^k}{2}$ which is strictly smaller than $k$ for $q$ large enough.
\end{proof}

\raggedright
\bibliographystyle{alpha}
\bibliography{all}

\end{document}